\documentclass[%
 reprint,
superscriptaddress,
amsmath,amssymb,
aps,
pre,
]{revtex4-1}

\usepackage{amsthm}
\usepackage{graphicx}
\usepackage{dcolumn}
\usepackage{bm}
\usepackage{xcolor}
\definecolor{link-color}{cmyk}{0.8 ,  0.3 ,  0. , 0}
\usepackage[colorlinks=true,%
            linkcolor   = {link-color},%
            citecolor   = {link-color},%
            urlcolor    = {link-color}]{hyperref}

\setcitestyle{numbers}

\newcommand{\Network}{G}
\newcommand{\Neighbors}{\mathcal{N}}
\newcommand{\Nodes}{V}
\newcommand{\Edges}{E}
\newcommand{\Arcs}{\mathcal{A}}
\newcommand{\edge}[2]{(#1,#2)}

\newcommand{\Adjacency}{A}

\renewcommand{\L}{\mathcal{L}}
\newcommand{\T}{{\!\top}}
\newcommand{\Q}{Q}

\DeclareMathOperator{\diag}{diag}
\DeclareMathOperator{\sgn}{sgn}

\newtheorem{theorem}{Theorem}
\newtheorem{definition}{Definition}
\newtheorem{lemma}{Lemma}

\graphicspath{ {../figures/} }

\begin{document}

\title{Partitioning signed networks}
\thanks{To appear as a chapter in ``Advances in Network Clustering and Blockmodeling'' edited by P. Doreian, V. Batagelj \& A. Ferligoj (Wiley, New York, 2018)}%

\author{Vincent Traag}
\email{v.a.traag@cwts.leidenuniv.nl}
\affiliation{Centre for Science and Technology Studies (CWTS), Leiden University}
\author{Patrick Doreian}
\affiliation{Faculty of Social Sciences, University of Ljubljana}
\affiliation{Department of Sociology, University of Pittsburgh}
\author{Andrej Mrvar}
\affiliation{Faculty of Social Sciences, University of Ljubljana}

\date{\today}

\begin{abstract}
\noindent Signed networks appear naturally in contexts where conflict or animosity is apparent.
In this book chapter we review some of the literature on signed networks, especially in the context of partitioning.
Most of the work is founded in what is known as structural balance theory.
We cover the basic mathematical principles of structural balance theory.
The theory yields a natural formulation for partitioning.
We briefly compare this to other partitioning approaches based on community detection.
Finally, we analyse an international network of alliances and conflicts and discuss the implications of our findings for structural balance theory.
\end{abstract}

\maketitle

\noindent We are concerned with signed networks, where each link is associated with either a positive ($+$) or negative sign ($-$).
More generally, weights $w_{ij}$ could be used. Although weights are often assumed to be positive, we explicitly allow them also to be negative.
For simplicity, we deal primarily with non-weighted networks, but most concepts used here can be adapted easily to the weighted case.

\section{Notation}

\noindent While we try to be as consistent as possible with the general notation used throughout this book,
we require some additional notation because signed networks have signs for arcs and edges.
We denote a directed signed network by $\Network=(\Nodes,\Arcs^-,\Arcs^+)$ where $\Arcs^-\subseteq \Nodes \times \Nodes$ are the negative links and $\Arcs^+ \subseteq \Nodes \times \Nodes$ the positive links. 
We assume that $\Arcs^- \cap \Arcs^+ = \emptyset$, so that no link is both positive and negative.
We exclude loops on nodes.
Many studied signed networks are directed.
Some are not, including the network we study here.
Similarly, an undirected signed network is denoted by $\Network=(\Nodes,\Edges^-,\Edges^+)$ where $\Edges^-\subseteq \Nodes \times \Nodes$ are the negative links and $\Edges^+ \subseteq \Nodes \times \Nodes$ the positive links.
As for the directed case, $\Edges^- \cap \Edges^+ = \emptyset$.

We present our initial discussion in terms of directed signed networks. However, if we restrict ourselves to undirected graphs, then $\edge{i}{j} \in \Edges^\pm$ is identical to $\edge{j}{i} \in \Edges^\pm$.
Also, we assume that there are no self-loops, i.e. no $\edge{i}{i}$ exists. 
For edges, the signs on them are symmetrical by definition.

We define the adjacency matrices $\Adjacency^{+}$ and $\Adjacency^{-}$. 
We set $\Adjacency^+_{ij} = 1$ whenever $\edge{i}{j} \in \Arcs^+$ and $\Adjacency^+_{ij} = 0$ otherwise.
Similarly, $\Adjacency^-_{ij} = 1$ whenever $\edge{i}{j} \in \Arcs^-$ and $\Adjacency^-_{ij} = 0$ otherwise.
We denote the \emph{signed} adjacency matrix $\Adjacency = \Adjacency^{+} - \Adjacency^{-}$.
This can be summarised as follows
\begin{equation}
  \Adjacency_{ij} = \begin{cases}
      -1 & \text{if~} (i,j) \in \Arcs^-, \\
       1 & \text{if~} (i,j) \in \Arcs^+, \\
       0 & \text{otherwise}.
    \end{cases}
\end{equation}
Note that we exclusively work with the \emph{signed} adjacency matrix in this chapter, and $\Adjacency$ should not be confused with the ordinary adjacency matrix.
The signed adjacency matrix for undirected networks is defined in a similar fashion.
For undirected networks the signed adjacency matrix is symmetric, and $\Adjacency = \Adjacency^\T$.

The neighbors of a node are those nodes to which it is connected.
The positive neighbors are $\Neighbors^+_v = \{ u \mid \edge{v}{u} \in \Edges^+\}$ and the negative neighbors similarly $\Neighbors^-_v = \{ u \mid \edge{v}{u} \in \Edges^-\}$, and all the neighbors are simply the union of both $\Neighbors(v) = \Neighbors^+(v) \cap \Neighbors^-(v)$.
The number of edges connected to a node is its degree. We distinguish between the positive degree $d^+_v = |\Neighbors^+_v|$, negative degree $d^-_v = |\Neighbors^-_v|$ and total degree $d_v = |N_v| =  d^+_v + d^-_v$. Similar formulations are possible for directed signed networks.

Blockmodeling, as a way of partitioning social networks, started with a clear
\emph{substantive} rationale expressed in terms of social roles~\cite{lorrain1971}. 
However, the availability of algorithms for partitioning (unsigned) networks~\cite{concor,burt1}, based on ideas of structural equivalence, led to a rather mechanical application to simply partition social networks with a subsequent {\it ad hoc} interpretation of what was identified. Such algorithms are indirect in the sense of having networks transformed to (dis)similarity measures for which partitioning methods are used. In contrast, a direct approach was proposed~\cite{Doreian2005} in which the network data are clustered directly. This allows the inclusion of substantive ideas within the rubric of pre-specification.

Consistent with this, the approach known as structural balance theory has a clear substantive foundation. We briefly review the basics of balance theory as it connects directly to partitioning signed social networks.
We then review some methods for partitioning networks in practice, and examine how they connect to balance theory.
Finally, we briefly explore how structural balance evolves through time in an empirical example of international alliances and conflict.

\section{Structural balance theory}

\noindent The basis of structural balance theory is founded on considerations of cognitive dissonance. 
\citet{Heider1946} focused on so-called p-o-x triplets, considering the relations between an actor (p), another actor (o) and some object (x) and claimed such triplets tend to be consistent in attitudes. 
For example, in this perspective, if someone (p) has a friend (o) who dislikes conservative philosophies (x), then p also tends to dislike conservative philosophies. 
This extends naturally to p-o-q triples for three actors denoted by p, o, and q. 
In the formulation involving three actors, well-known claims such as ``a friend of a friend is a friend'';  ``an enemy of a friend is an enemy'';  ``a friend of an enemy is an enemy''; and ``an enemy of an enemy is a friend'' are thought to hold. 
The notion of balance from \citet{Heider1946} was further formalized, and extended to an arbitrary number of persons or objects by~\citet{Cartwright1956}.
They modeled relations between persons as a graph where nodes are persons and the relations between them links in the graph. The four possible triads for the undirected case are shown in Figure~\ref{fig:struct_balance}.

\begin{figure}[t]
  \begin{center}
    \includegraphics{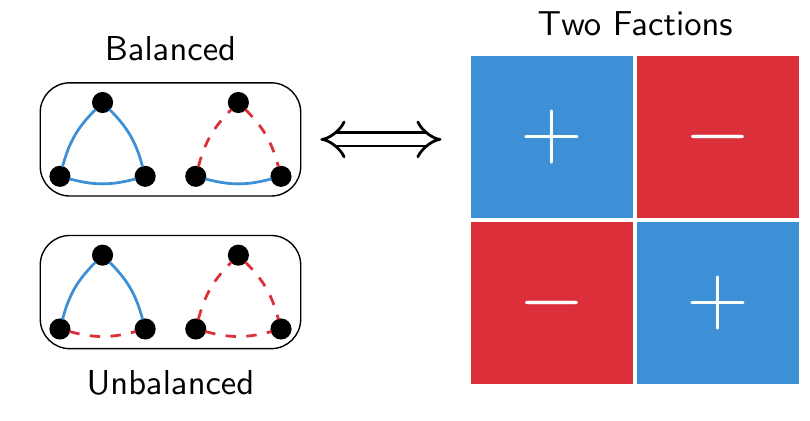}
  \end{center}
  \caption{\textbf{Structural balance.} 
    There are four possible configurations for having positive or negative links between three nodes (a triad). 
    These are demonstrated on the left, where a solid line represents a positive link and a dashed line represents a negative link. 
    The upper two triads are structurally balanced because the product of their signs is positive.
    Similarly, the lower two triads are not structurally balanced because the product of their signs is negative.
    If all triads (in a complete network) are structurally balanced, the network can be partitioned in two factions such that they are internally positively linked, with negative links between the two factions, as illustrated on the right.
  }
  \label{fig:struct_balance}
\end{figure}

For the remainder of the chapter, we restrict ourselves to undirected graphs.
We first focus on complete graphs, where all links are present (excluding self-loops).
Following~\cite{Cartwright1956}, we provide the following definition.

\begin{definition}
  A triad\index{triad}{} $i,j,k$ is called balanced whenever the product
  \begin{equation}
    \Adjacency_{ij}\Adjacency_{jk}\Adjacency_{ki} = 1.
    \label{equ:balanced_triad}
  \end{equation}
  A complete signed graph $\Network$ is structurally balanced whenever all triads are balanced.
\end{definition}

Of the four possible triads, two are balanced ($+++$ and $+--$) and two are unbalanced ($++-$ and $---$) according to this definition (see also Fig.~\ref{fig:struct_balance}).

\citet{Harary1953} proved that if the graph $\Network$ is structurally balanced, then it can be partitioned in two clusters such that there are only positive links within each cluster and negative links between them.
\citet{Cartwright1956} called this observation the structure theorem, and \citet{Doreian1996a} called it the \emph{first} structure theorem . 
\begin{theorem}[Structure theorem, \cite{Harary1953}]
  Let $\Network = (\Nodes,\Edges^+,\Edges^-)$ be a complete signed graph. If and only if $\Network$ is balanced can $\Nodes$ be partitioned into two disjoint subsets $\Nodes_1$ and $\Nodes_2$ such that a positive edge $e \in \Edges^+$ either in $\Nodes_1 \times \Nodes_1$ or $\Nodes_2 \times \Nodes_2$ while a negative edge $e \in \Edges^-$ falls in $\Nodes_1 \times \Nodes_2$.
\end{theorem}
\begin{proof}
	Assume $\Network$ is balanced.
	Consider some node $v \in \Nodes$ and set $\Nodes_1 = v \cup \Neighbors^+(v)$ as well as the set $\Nodes_2 = \Nodes \setminus \Nodes_1$.
	Consider an edge $\edge{u}{w} \in \Nodes_2 \times \Nodes_2$. 
  Then $\edge{u}{v} \in \Edges^-$ and $\edge{w}{v} \in \Edges^-$ by definition of $\Nodes_2$ so that $\edge{u}{w} \in \Edges^+$ by structural balance.
  Hence all edges in $\Nodes_2$ are positive.
  Similarly, any edge $\edge{u}{w} \in \Nodes_1 \times \Nodes_1$ is positive.
	Hence, we can partition $\Nodes$ into the stated disjoint sets $\Nodes_1$ and $\Nodes_2$.
	In reverse, any triad is easily seen to be balanced if $\Nodes$ is partitioned as stated in the theorem.
\end{proof}

While the above is limited to complete graphs, it can be generalized to incomplete graphs.
For this we first need to introduce another definition for structural balance:

\begin{definition}[Structural Balance]
  Let $\Network = (\Nodes,\Edges^+,\Edges^-)$ be a signed graph and $\Adjacency$ the signed adjacency matrix.
  Let $C = v_1v_2\ldots v_kv_1$ be a cycle consisting of nodes $v_i$ with $v_{k+1} = v_1$. 
	Then the cycle $C$
  is called balanced whenever 
  \begin{equation}
    \sgn(C) := \prod_{i=1}^k \Adjacency_{v_i v_{i+1}} = 1.
  \end{equation}
	
  A signed graph $\Network$ is called balanced if all its cycles $C$ are balanced.
\end{definition}

Stated differently, $\sgn(C)$  is the sign of the cycle which is balanced if its sign is positive.
If a cycle contains $m^-$ negative edges, then $\sgn(C) = (-1)^{m^-}$.
In other words, a cycle is balanced if it contains an even number of negative links.
Note that for a cycle of length three, this coincides exactly with the definition of a balanced triad.

The sign of a cycle can be decomposed in the sign of subcycles if the cycle has a \emph{chord}: an edge between two nodes of the cycle (see Fig.~\ref{fig:sb_cycles}).
\begin{theorem}
	Let $C = v_1v_2\ldots v_kv_1$ be a cycle with a chord between nodes $v_1$ and $v_r$ in $C$. 
	Then let $C_1 = v_1v_2\ldots v_rv_1$ and $C_2 = v_1v_k\ldots v_r v_1$ be the induced subcycles. 
	Then $\sgn(C) = \sgn(C_1)\sgn(C_2)$.
\end{theorem}
\begin{proof}
	We denote by $m^-_1$ the number of negative links of $C_1$ and similarly $m^-_2$ for $C_2$ and $m^-$ for $C$.
  Suppose that the link $\edge{v_1}{v_r}$ is not a negative link.
	Then the number of negative links in $C$ is $m^- = m^-_1 + m^-_2$ so that $\sgn(C) = (-1)^{m^-} = (-1)^{m^-_1}(-1)^{m^-_2} = \sgn(C_1) \sgn(C_2)$.
  Suppose that $\edge{v_1}{v_r}$ is a negative link.
	Then $m^- = (m^-_1 - 1) + (m^-_2 - 1)$ so that $\sgn(C) = (-1)^{m^-} = (-1)^{m^-_1}(-1)^{m^-_2}(-1)^{-2} = \sgn(C_1)\sgn(C_2)$.
\end{proof}
In other words, it is not necessary to determine the structural balance of all cycles, and we can restrict ourselves to the balance of chordless cycles.
In fact, this statement can be made stronger, and holds for any combination of cycles.
With a combination of cycles, we mean the symmetric difference of the edges of the two cycles.
To define this properly it is more convenient to denote a cycle by the set of its edges (in no particular order).
That is, we define a cycle $C = \{ e_1, e_2, \ldots, e_k \}$ where the edges form a cycle, i.e. the subgraph of $C$ is a cycle.

\begin{definition}
Let $C_1 = \{ e_1, e_2, \ldots, e_k \}$ and $C_2 = \{f_1, f_2, \ldots, f_k \}$
be two cycles. Then we define the \emph{symmetric difference} as
	\begin{align}
		C_1 \triangle C_2 = (C_1 \cup C_2) \setminus (C_1 \cap C_2).
	\end{align}
  We also refer to this as the \emph{combination} of two cycles.
\end{definition}
	
Note that the combination of two cycles may actually be a set of multiple edge-disjoint cycles.

Now we can prove the stronger statement on the combination of cycles.
\begin{theorem}
	Let $C_1 = \{ e_1, e_2, \ldots, e_k \}$ and $C_2 = \{f_1, f_2, \ldots, f_k \}$ be two cycles.
	If $C = C_1 \triangle C_2$ then $\sgn(C) = \sgn(C_1)\sgn(C_2)$.
\end{theorem}
\begin{proof}
	Let us denote the number of negative links in a set $C$ by $m^-(C) = | C \cap E^- | $.
	Let $S = C_1 \cup C_2$ be the union of the two cycles, and $T = C_1 \cap C_2$ be the overlap of the two cycles.
	Then $m^-(S) + m^-(T) = m^-(C_1) + m^-(C_2)$, and $m^-(C) = m^-(S) - m^-(T) = m^-(C_1) + m^-(C_2) - 2 m^-(T)$.
	Hence
	\begin{align}
		\sgn(C) &= (-1)^{m^-(C)} \\
						&= (-1)^{m^-(C_1) + m^-(C_2) - 2 m^-(T)} \\
						&= (-1)^{m^-(C_1)} (-1)^{m^-(C_2)} (-1)^{2 m^-(T)} \\
						&= (-1)^{m^-(C_1)} (-1)^{m^-(C_2)} \\
						&= \sgn(C_1) \sgn(C_2)
	\end{align}
	since $(-1)^{2 m} = 1$ for any integer $m$.
\end{proof}

In other words, if we know the balance of some limited number of cycles, we can determine the balance of all cycles.
These `limited number of cycles' are called the \emph{fundamental cycles}. Any cycle can be obtained as a combination of two (or  more) fundamental cycles.
This implies that if all fundamental cycles are balanced, then the graph as a whole is balanced.
We do not consider fundamental cycles in more detail, but this notion underlies the technique by~\citet{Altafini2012} which we consider in section~\ref{sec:switching}.

\begin{figure}[t]
  \begin{center}
    \includegraphics{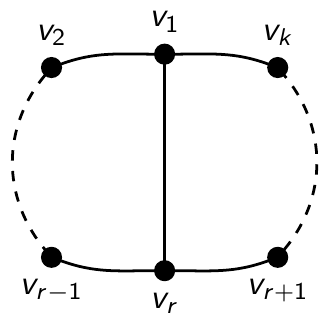}
  \end{center}
  \caption{\textbf{Chords and cycles.} 
  This illustrates a cycle $v_1 \ldots v_k$ with a chord between nodes $v_1$ and $v_r$. 
  There are two subcycles: one following the left path and the other following the right path in the illustration.
  These two subcycles have a single common edge: $v_1v_r$.
  The sign of the large cycle is then the product of the sign of the two subcycles.
  }
  \label{fig:sb_cycles}
\end{figure}

Similar to the sign of cycles, we can define the sign of a path.

\begin{definition}
 Let $P = v_1v_2\ldots v_k$ be a path in a signed graph $\Network$ with signed adjacency matrix $\Adjacency$. 
	The sign of the path $P$ is then defined as
  \begin{equation}
    \sgn(P) := \prod_{i=1}^{k-1} \Adjacency_{v_i v_{i+1}}.
  \end{equation}
\end{definition}

Paths in signed networks are either positive or negative.
A cycle can be decomposed in two paths so the sign of a cycle is the product of the sign of the two paths.
Hence, a cycle is balanced if the two paths have the same sign.

As before, the graph $\Network$ can be partitioned in two clusters with positive links within clusters and negative links between clusters.
\begin{theorem}[Structure theorem, \cite{Harary1953}]
  Let $\Network = (\Nodes,\Edges^+,\Edges^-)$ be a connected signed graph and $\Adjacency$ the signed adjacency matrix.
  Then $\Network$ is balanced if and only if $\Network$ can be partitioned into two disjoint subsets $\Nodes_1$ and $\Nodes_2$ such that a positive edge $e \in \Edges^+$ either in $\Nodes_1 \times \Nodes_1$ or $\Nodes_2 \times \Nodes_2$ while a negative edge $e \in \Edges^-$ falls in $\Nodes_1 \times \Nodes_2$.
  \label{thm:structure}
\end{theorem}
\begin{proof}
  First, assume $\Network$ is balanced. 
	Then select any $v \in \Nodes$ and set $\Nodes_1 = \{u | \sgn(u-v \text{~path}) = 1 \}$, that is, all the nodes that can be reached through a positive path. 
	Define $\Nodes_2 = \Nodes \setminus \Nodes_1$. 
	Let $e = \edge{u}{w} \in \Edges^-$. 
	Suppose $e \in \Nodes_1 \times \Nodes_1$.
	By construction of $\Nodes_1$, then both $u$ and $w$ have a positive path to $v$, so that the path $u-w$ through $v$ is also positive.
	But if $\edge{u}{w}$ is negative, it would be contained in a negative cycle, contradicting balance.
	Hence $e \notin \Nodes_1 \times \Nodes_1$.
	Similarly, suppose that $e \in \Nodes_2 \times \Nodes_2$.
	Then both the $u-v$ path and the $w-v$ path is negative (otherwise $u$ and $w$ would be in $\Nodes_1$).
	The $u-w$ path through $v$ is then positive since the product of the two negative paths is positive. 
  Again, since $\edge{u}{w} \in \Edges^-$ it contradicts balance. 
  Hence, all negative edges lie between $\Nodes_1$ and $\Nodes_2$.
  Finally, let $e = \edge{u}{w} \in \Edges^+$ with $u \in \Nodes_1$ and $w \in \Nodes_2$.
  Then there is a positive $u-v$ path and a negative $w-v$ path, so that the $u-w$ path through $v$ is negative, which combined with the positive edge $\edge{u}{w}$ leads to a negative cycle, contradicting balance.
	Hence, positive edges lie within $\Nodes_1$ and $\Nodes_2$. 
  We conclude that if $\Network$ is balanced, it can be partitioned as stated.
	Vice-versa, suppose $\Network$ can be partitioned into the two states subsets $\Nodes_1$ and $\Nodes_2$. 
	Let $C$ be a cycle. 
	If $C$ is contained within $\Nodes_1$ or $\Nodes_2$ it is completely positive, so that $\sgn(C) = 1$. 
	Suppose $C$ has some node $u \in \Nodes_1$ and $v \in \Nodes_2$. 
	Then any $u-v$ path contains an odd number of negative links, and is hence negative, so that the cycle $C$ is positive.
	Hence, all cycles are balanced, and so $\Network$ is balanced.
\end{proof}

\subsection{Weak structural balance}

Classical structural balance theory predicts that a balanced network can be partitioned into two clusters.
However, as suggested by \citet{Davis1967} and \citet{Cartwright1968}, we can generalize this notion of structural balance by redefining the notion of an unbalanced triad or cycle.
Consider for example the (unbalanced) triad with three negative links.
The three nodes can be partitioned into three clusters: trivially, all links between clusters are negative and all positive links are within clusters.
There is a simple characterization of networks that can be partitioned in such a way: no cycle can contain exactly one negative link.
\citet{Davis1967} established this only for complete graphs, and \citet{Cartwright1968} extended this to sparse graphs.
We call signed networks with this property \emph{weakly} structurally balanced (or weakly balanced).

\begin{definition}
  A cycle $C = v_1v_2\ldots v_kv_1$ is termed weakly balanced if it does not contain exactly a single negative link. 
	A signed graph $\Network$ is called weakly balanced if all its cycles $C$ are weakly balanced.
\end{definition}

Following this, we can call the previous definition \emph{strong} structural balance.
Any graph that is strongly structurally balanced is also weakly structurally balanced: a cycle with positive sign must contain an even number of negative links. 
It cannot have exactly one. 
The reverse does not hold: a weakly structurally balanced cycle can have three negative links which is not allowed in strong structural balance.

\begin{lemma}
  Let $C = v_1v_2\ldots v_kv_1$ be a cycle with a chord between nodes $v_1$ and $v_r$ in $C$. 
	Then let $C_1 = v_1v_2\ldots v_rv_1$ and $C_2 = v_1v_k\ldots v_r v_1$ be the induced subcycles. 
	Then $C$ is weakly balanced if $C_1$ and $C_2$ are weakly balanced. 
  \label{lem:weak_bal:chordless_cycles}
\end{lemma}

\begin{proof}
	We denote by $m^-_1 \neq 1$ the number of negative links of $C_1$ and, similarly, $m^-_2 \neq 1$ for $C_2$ and $m^-$ for $C$.
  Suppose that the link $\edge{v_1}{v_r}$ is not a negative link, then the number of negative links in $C$ is $m^- = m^-_1 + m^-_2 \neq 1$ implying $C$ is weakly balanced. 
  Suppose that $\edge{v_1}{v_r}$ is a negative link.
	Then both $m^-_1 \geq 2$ and $m^-_2 \geq 2$, and $m^- = (m^-_1 - 1) + (m^-_2 - 1) \geq 2$ so that $C$ is weakly balanced.
\end{proof}

The inverse does not hold. 
This can readily be seen by considering an all-positive cycle with a single negative chord.
The all-positive cycle, clearly, is weakly balanced, but the induced sub-cycles contain exactly one single negative link, and are therefore not weakly balanced.
The theorem on chordless cycles for weak balance is hence a weaker statement than the corresponding theorem for strong structural balance.
Nonetheless, we can still limit ourselves to considering chordless cycles for determining whether a graph is weakly structurally balanced.
\begin{theorem}
	Let $\Network$ be a signed network. Then $\Network$ is weakly structurally balanced if and only if all chordless cycles are weakly structurally balanced.
\end{theorem}
\begin{proof}
	If $\Network$ is weakly balanced, all cycles are balanced, so that trivially all chordless cycles are balanced.
	Vice versa, assume all chordless cycles are weakly balanced.
	We use induction on $|C|$.
  All chordless cycles $C$ are balanced by assumption, providing our inductive base for $|C| = 3$ (because triads are chordless by definition).
	Assume all cycles with $|C| < r$ are balanced, then consider cycle $C$ of length $r$.
	If $C$ contains a chord, we can separate $C$ in cycles $C_1$ and $C_2$, which are balanced by our inductive assumption.
  Then, by Lemma~\ref{lem:weak_bal:chordless_cycles} cycle $C$ is balanced.
	Hence, all cycles are weakly balanced.
\end{proof}
To determine whether a graph is weakly structurally balanced, we need only to consider the chordless cycles rather than all cycles.
Computationally, this is important.

Similar to strong structural balance, we can partition a weakly structurally balanced graph, but now in possibly more than two clusters.
This is called the \textit{second} structure theorem by \citet{Doreian1996a}.
\begin{theorem}[Clusterability theorem, \cite{Cartwright1968}]
  Let $\Network = (\Nodes,\Edges^+,\Edges^-)$ be a connected signed graph.
  Then $\Network$ is weakly structurally balanced if and only if $\Network$ can be partitioned into disjoint subsets $\Nodes_1, \Nodes_2, \ldots, \Nodes_r$ such that a positive edge $e \in \Edges^+$ falls in $\Nodes_c \times \Nodes_c$ while a negative edge $e \in \Edges^-$ falls in $\Nodes_c \times \Nodes_d$ for $c \neq d$.
  \label{thm:clusterability}
\end{theorem}
\begin{proof}
  Suppose $\Network$ is weakly balanced. 
	Let $\Network^+ = (\Nodes,\Edges^+)$ be the positive part of the signed graph, and let the clusters be defined by the connected components of $\Network^+$. 
	Any positive edge then clearly cannot fall between clusters, because different connected components cannot be connected through a positive link.
	Consider then some negative link $\edge{u}{v} \in \Edges^-$.
	Suppose that $u$ and $v$ are both in some $\Nodes_c$.
	Then there exists a positive $u-v$ path because they are in the same component, thus yielding a cycle with exactly a single negative link, contradicting weak balance.
	Hence, any negative link falls between clusters.
	Vice versa, suppose $\Network$ is split into clusters as stated in the theorem. 
	Any cycle completely contained within a cluster has only positive links.
	Consider a cycle through $u$ and $v$ where $u \in \Nodes_c$ and $v \in \Nodes_d$, $d \neq c$.
	Then any path between $u$ and $v$ must contain at least a single negative link, so that any cycle must contain at least two negative links.
\end{proof}
It is easy to see when a complete signed graph is weakly structurally balanced: it must not contain the $++-$ triad. 

In summary, signed networks which are strongly structurally balanced can be partitioned in two clusters.
Signed networks which are weakly structurally balanced can be partitioned in multiple clusters. 
Clearly, all signed networks which are strongly structurally balanced are also weakly structurally balanced, but not vice versa. 
One obvious question is whether strong or weak structural balance is more realistic. 
This led to partitioning signed networks, which we will examine in the next section.

\section{Partitioning}

\noindent The previous section introduced the general idea and structure theorems for structural balance.
However, these conditions are rather strict: \emph{no} cycle can exist with an odd number of negative links (strong balance), or a single negative link (weak balance).
Empirically, this is rather unrealistic to achieve exactly, but we might come close.
This was suggested by \citet{Cartwright1956}, when introducing the notion of structural balance, who suggested counting the number of cycles that are balanced and measuring the proportion of balanced cycles, termed the \emph{degree of balance}:
\begin{equation}
	b(\Network) = \frac{c^+(\Network)}{c(\Network)}
\label{eq:degree_of_balance}
\end{equation}
where $c^+(\Network)$ is the number of balanced cycles and $c(\Network)$ is the total number of cycles.
This measure is used infrequently, because it is computationally intensive to list all cycles~\cite{Hummon1995}.
The number of cycles in a graph increases exponentially with its size.
Depending on the so-called cyclomatic number, $\mu = m - n + 1$, there are between $\mu$ and $2^\mu$ cycles~\cite{Volkmann1996}, which \citet{Harary1959} also uses to define bounds on the degree of balance.
However, this number provides little insight into the structure of the network.

A more useful measure was suggested by \citet{Harary1959}: the smallest number of  ties to be deleted in order to make the network (weakly) balanced. This is the same as the number of ties whose reversal of signs leads to a balanced network.
This is known as the \emph{line index of imbalance}.
Computing the line index of imbalance is computationally intensive as it is an NP-hard problem.
Initially the definition was restricted to strong structural balance.
\citet{Doreian1996a} were the first to introduce this in the context of clustering for weak structural balance.

\subsection{Strong structural balance}

Given a partition into two subsets, $\Nodes_1$ and $\Nodes_2$, we can measure the number of edges that are in conflict with structural balance.
The number of negative edges within $\Nodes_1$ are
\begin{equation}
	C^-(\Nodes_1) = \frac{1}{2}\sum_{i \in \Nodes_1, j \in \Nodes_1} \Adjacency^-_{ij} \\
\end{equation}
and similarly so for $\Nodes_2$, while the positive edges between $\Nodes_1$ and $\Nodes_2$ are
\begin{equation}
	C^+(\Nodes_1,\Nodes_2) = \sum_{i \in \Nodes_1, j \in \Nodes_2} \Adjacency^+_{ij} \\
\end{equation}
so that the total number of edges inconsistent with structural balance for a partition into $\Nodes_1$ and $\Nodes_2$ is
\begin{equation}
	C(\Nodes_1,\Nodes_2) = C^-(\Nodes_1) + C^-(\Nodes_2) + C^+(\Nodes_1, \Nodes_2).
\end{equation}
This is the \emph{line index of imbalance} mentioned earlier.
A graph $\Network$ is then structurally balanced if and only if the minimum line index of imbalance is zero.

\subsubsection{Spectral theory}

Given a partition into $\Nodes_1$ and $\Nodes_2$, let $x_i = 1$ if $i \in \Nodes_1$ and $x_i = -1$ if $i \in \Nodes_2$. 
Then for an edge $\edge{i}{j}$, if $x_i = x_j$ then $x_i \Adjacency_{ij} x_j = \Adjacency_{ij}$, while for $x_i \neq x_j$ we have $x_i \Adjacency_{ij} x_j = -\Adjacency_{ij}$.
Hence
\begin{align}
	x^\T \Adjacency x =& \sum_{x_i = x_j} (\Adjacency^+_{ij} - \Adjacency^-_{ij}) + \sum_{x_i \neq x_j} (\Adjacency^-_{ij} - \Adjacency^+_{ij}) \\
					 =& 2m - \sum_{x_i = x_j} (\Adjacency^+_{ij} + \Adjacency^-_{ij}) - \sum_{x_i \neq x_j} (\Adjacency^+_{ij} + \Adjacency^-_{ij}) \nonumber \\
            & + \sum_{x_i = x_j} (\Adjacency^+_{ij} - \Adjacency^-_{ij}) + \sum_{x_i \neq x_j} (\Adjacency^-_{ij} - \Adjacency^+_{ij}) \\
					=& 2m - 2\sum_{x_i = x_j} \Adjacency^-_{ij} - 2\sum_{x_j \neq x_j} \Adjacency^+_{ij}
\end{align}
So that $x^\T \Adjacency x = 2(m - C(\Nodes_1, \Nodes_2))$ gives (twice) the number of edges that are consistent with balance, the inverse of the line index of imbalance.
Note that this also implies that if $x_i$ is the partition corresponding to structural balance, than $x_i \Adjacency_{ij} x_j > 0$ for all $i$, $j$.

\begin{theorem}
  \label{thm:spectral}
  Let $\Network$ be a connected signed graph and let $u$ be the dominant eigenvector of the signed adjacency matrix $\Adjacency$. 
  Then $\Network$ is balanced if and only if $\Nodes_1 = \{i \in \Nodes | u_i \geq 0\}$ and $\Nodes_2 = \Nodes\setminus \Nodes_1$ defines the split into two clusters as in Theorem~\ref{thm:structure}.
\end{theorem}

\begin{proof}
  If the split defines a correct partition, then obviously $\Network$ is balanced (Theorem~\ref{thm:structure}). 
  In reverse, suppose $\Network$ is balanced. 
  Let $u$ be the dominant eigenvector. 
  Suppose that $u_i \Adjacency_{ij} u_j < 0$ for some $i,j$. 
  Then let $x$ be another vector with $|x_i| = |u_i|$ for all $i$ and $x_i \Adjacency_{ij} x_j \geq 0$ for all $i,j$, which is possible by structural balance of $\Network$. 
  Then $\|x\| = \|u\|$ and
  \begin{align}
    u^\T \Adjacency u &= \sum_{ij} u_i \Adjacency_{ij} u_j \\
             &< \sum_{ij} |u_i \Adjacency_{ij} u_j | \\
             &= \sum_{ij} |x_i \Adjacency_{ij} x_j | \\
             &= \sum_{ij} x_i \Adjacency_{ij} x_j = x^\T \Adjacency x,
  \end{align}
  which contradicts the fact that $u$ is the dominant eigenvector. Hence, $u_i \Adjacency_{ij} u_j \geq 0$ for all $i,j$ and it defines a correct
  partition.
\end{proof}

The vector space constrained to $|x_i| = 1$ is rather difficult to optimize.
Taking general vectors with $\|x\| = 1$, the dominant eigenvector $x$ maximizes this and the largest eigenvalue of the adjacency matrix $\lambda_n(\Adjacency)$ gives a lower bound of the line index of imbalance.

\citet{Kunegis} suggest using the signed Laplacian~\citep{Hou2003} for measuring structural balance.
It is defined as
\begin{equation}
	\L = D - \Adjacency
\label{eq:laplacian}
\end{equation}
where $\Adjacency$ is the signed adjacency matrix as defined earlier and $D = \diag(d_1, \ldots, d_n)$ the diagonal matrix of total degrees.
The rows of $\L$ sum to twice the negative degrees $2 ( d^-_1, \ldots, d^-_n)$ because $\sum_{j} \Adjacency_{ij} = d^+_i - d^-_i$, so that $(d^+_i + d^-_i) - (d^+_i - d^-_i) = 2d^-_i$.
Furthermore, the Laplacian is positive-semidefinite, i.e. $x^\T \L x \geq 0$ for all $x$.
We can show this as follows.
Writing out, we obtain
\begin{align}
   & x^\T \L x \nonumber \\
  =& \sum_{ij} x_i \L_{ij} x_j \\
  =& \sum_{ij} x_i \delta_{ij} d_i x_j - \sum_{ij} x_i \Adjacency_{ij} x_j. \\
\intertext{Since $d_i = \sum_{ij} |\Adjacency_{ij}|$, we can write $\sum_{ij} x_i \delta_{ij} d_i x_j = \sum_{ij} |\Adjacency_{ij}| x_i^2$ and obtain}
	=& \sum_{ij} |\Adjacency_{ij}| x_i^2 - \sum_{ij} x_i \Adjacency_{ij} x_j. \\
\intertext{Clearly $\sum_{ij} |\Adjacency_{ij}| x_i^2 = \sum_{ij} |\Adjacency_{ij}| x_j^2$ so that we get}
	=& \frac{1}{2} \left(\sum_{ij} |\Adjacency_{ij}| x_i^2 + \sum_{ij} |\Adjacency_{ij}| x_j^2 - 2\sum_{ij} x_i \Adjacency_{ij} x_j\right), \\
\intertext{which can be nicely expressed as a square (because $\Adjacency^2_{ij} = |\Adjacency_{ij}|$ and $|\Adjacency_{ij}|^2 = |\Adjacency_{ij}|$)}
	=& \frac{1}{2} \sum_{ij} |\Adjacency_{ij}| (x_i - \Adjacency_{ij} x_j)^2 \geq 0.
\end{align}

Now suppose $\Network$ is strongly balanced so that we can partition  the nodes into $\Nodes_1$ and $\Nodes_2$ without violating balance.
Let $x_i = 1$ if $i \in \Nodes_1$ and $x_i = -1$ if $i \in \Nodes_2$. 
Then for any edge $\edge{i}{j}$, if $x_i = x_j$ then by strong balance $\Adjacency_{ij} = 1$, while if $x_i = -x_j$ we have $\Adjacency_{ij} = -1$.
Hence $|\Adjacency_{ij}|(x_i - \Adjacency_{ij} x_j)^2 = 0$ and the smallest eigenvalue of the Laplacian is $0$.
Vice versa, if the Laplacian is $0$, $\Network$ is balanced: the term $|\Adjacency_{ij}|(x_i - \Adjacency_{ij} x_j)^2 = 0$ can only be $0$ for all $ij$ if $\Adjacency$ is balanced.

More generally, given a partition into $\Nodes_1$ and $\Nodes_2$, let $x_i = 1$ if $i \in \Nodes_1$ and $x_i = -1$ if $i \in \Nodes_2$. 
Then for an edge $\edge{i}{j}$, if $x_i = x_j$ then $|\Adjacency_{ij}|(x_i - \Adjacency_{ij} x_j)^2 = 4A^-_{ij}$ while if $x_i = -x_j$ then $|\Adjacency_{ij}|(x_i - \Adjacency_{ij} x_j)^2 = 4A^+_{ij}$.
Hence,
\begin{align}
	x^\T \L x &= \frac{1}{2} \sum_{ij} |\Adjacency_{ij}|(x_i - \Adjacency_{ij} x_j)^2 \\
           &= \frac{1}{2} \left(\sum_{x_i = x_j} 4A^-_{ij} + \sum_{x_i \neq x_j} 4A^+_{ij}\right) \\
	  &= 2 C(\Nodes_1,\Nodes_2)
\end{align}
and the vector $x$ gives (twice) the line index of imbalance.

The vector space constrained to $|x_i| = 1$ is rather difficult to optimize.
Taking general vectors with $\|x\| = 1$, the minimal eigenvector $x=u$ minimizes $x^\T \L x$.
Consequentially, the smallest eigenvalue of the Laplacian $\lambda_1(\L)$ gives a lower bound of the line index of imbalance, as $x^\T \L x \geq u^\T \L u$ where $u$ is the smallest eigenvector.
The partition induced by $u$ however, taking $x = \sgn(u)$, i.e. $x_i = \sgn(u_i)$, gives an upper bound, as the minimum index of imbalance is at most the index of an actual partition. 
Hence, we obtain
\begin{equation}
	\lambda_1(\L) \leq 2 C(\Nodes_1, \Nodes_2) \leq \tilde{\lambda}_1(\L)
\end{equation}
where $\tilde{\lambda}_1(\L) = \sgn(u)^\T \L \sgn(u)$.

We thus obtain the identity that $x^\T \Adjacency x = 2m - x^\T \L x$ and that maximizing $x^\T \Adjacency x$ is equivalent to minimizing $x^\T \L x$.
However, the eigenvectors of the adjacency matrix and the Laplacian are, in general, not identical.
In the case of balanced graphs though, the largest eigenvector of the adjacency matrix and the smallest eigenvector of the Laplacian provide identical information; the partition into $\Nodes_1$ and $\Nodes_2$.

\subsubsection{Switching}
\label{sec:switching}
One interesting observation in signed graph theory is that we can change the sign of some links without affecting balance.
More precisely, we can switch the signs of edges across a cut without changing structural balance.
Switching signs was introduced originally by \citet{Abelson1958} who used it to calculate the line index of imbalance (although they called it the ``complexity'' of a signed graph).
This was later used by \citet{Zaslavsky1982} in a formal graph-theoretical setting.
More recently, \citet{Iacono2010} use sign switches in an algorithm for calculating the line index of imbalance.

\begin{figure}[t]
  \centering
  \includegraphics[width=0.8\linewidth]{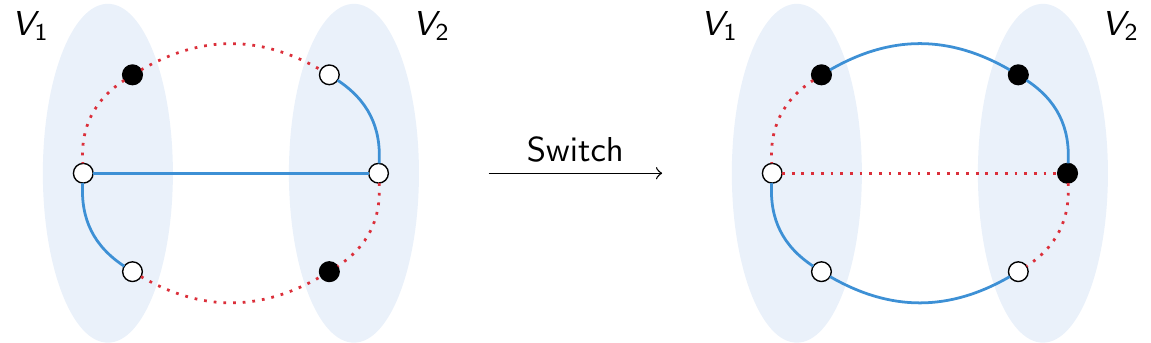}
  \caption{\textbf{Switching}
    This illustrates how switching works. 
    In the graph on the left, there are three edges crossing the partition into $\Nodes_1$ and $\Nodes_2$; two negative (the dashed lines) and one positive (the solid line).
    When we switch according to the partition $\Nodes_1$ and $\Nodes_2$, it implies that we switch the signs of the edges crossing the partition, but leave all the other signs unchanged.
    This is illustrated in the graph on the right.
    All cycles keep the same sign after the switching.
    In this case this reduces the number of negative links, and simplifies finding the balanced partition.
    The balanced partition is indicated by black and white nodes in both cases.
    In the right, the black and white are reversed for $\Nodes_2$, corresponding to the switching of the balanced partition by $\Nodes_1$ and $\Nodes_2$ as explained in Theorem~\ref{thm:switch_partition}.
  }
  \label{fig:switching}
\end{figure}

\begin{definition}[Switching]
	Let $\Network = (\Nodes, \Edges^+, \Edges^-)$ be a signed graph with signed adjacency matrix $\Adjacency$ and let $\Nodes_1$ and $\Nodes_2$ be a partition of $\Nodes$.
	Then let $s_i = 1$ if $i \in \Nodes_1$ and $s_i = -1$ if $i \in \Nodes_2$, with $S = \diag(s)$ and define $\hat{\Adjacency} = S \Adjacency S$ so that $\hat{\Adjacency}_{ij} = s_i \Adjacency_{ij} s_j$.
	Then the graph $\hat{\Network}$ defined by $\hat{\Adjacency}$ is called a switching of $\Network$ defined by the partition $\Nodes_1$ and $\Nodes_2$.
\end{definition}

Hence, for a link $\edge{i}{j}$ with $i \in \Nodes_1$ and $j \in \Nodes_2$, then $\hat{\Adjacency}_{ij} = -\Adjacency_{ij}$, while if both $i, j \in \Nodes_1$ (or $i,j \in \Nodes_2$), $\hat{\Adjacency}_{ij} = \Adjacency_{ij}$.
In other words, switching means we invert the signs of links across the cut by the partition $\Nodes_1$ and $\Nodes_2$, as illustrated in Fig.~\ref{fig:switching}.
Most importantly, any switching preserves the balance of any cycle.

\begin{theorem}
	Let $\Network = (\Nodes, \Edges^+, \Edges^-)$ be a signed graph and let $\hat{\Network}$ be a switched signed graph.
  Denote by $\sgn_\Network(C)$ the sign of some cycle $C$ with respect to $\Network$.Then for any cycle $C$, $\sgn_\Network(C) = \sgn_{\hat{\Network}}(C)$.
\end{theorem}
\begin{proof}
	Let $C$ be a cycle and let $\Nodes_1$ and $\Nodes_2$ be a partition of $\Network$.
	Let $m^\pm_\text{cut}$ be the number of positive/negative links across the cut between $\Nodes_1$ and $\Nodes_2$ in $\Network$ and $m^\pm_\text{within}$ the number of positive/negative links within $\Nodes_1$ or $\Nodes_2$.
	Then $\sgn_\Network(C) = (-1)^{m^-_\text{cut} + m^-_\text{within}}$.
	Hence there are $\hat{m}^-_\text{cut} = m^+_\text{cut}$ and $\hat{m}^+_\text{cut} = m^-_\text{cut}$ across the cut in $\hat{\Network}$, while $\hat{m}^\pm_\text{within} = m^\pm_\text{within}$.
  By definition $m^+_\text{cut} + m^-_\text{cut}$ is even since any cycle must cross $\Nodes_1$ and $\Nodes_2$ an even number of times.
  In other words, $(-1)^{m^+_\text{cut} + m^-_\text{cut}} = 1$ and hence $(-1)^{m^+_\text{cut}} = (-1)^{m^-_\text{cut}}$.
	We thus obtain
	\begin{align}
    \sgn_{\hat{\Network}}(C) &= (-1)^{\hat{m}^-_\text{cut} + \hat{m}^-_\text{within}} \\
                        &= (-1)^{m^+_\text{cut}}(-1)^{m^-_\text{within}} \\
                        &= (-1)^{m^-_\text{cut}}(-1)^{m^-_\text{within}} \\
                        &= \sgn_\Network(C)
	\end{align}
\end{proof}

Recall that the line index of imbalance is the minimum number of signs that would need to be changed to make the graph structurally balanced.
So, if the balance of the cycles does not change, then neither would the minimum number of sign changes required, and hence the line index of imbalance remains the same.

\begin{theorem}
	Let $\sigma_i = \{-1, 1\}$ be a partition of $\Network$ and let $S$ be a switching of $\Network$.
  Then the switched partition $\hat{\sigma} = \sigma S$ has the same imbalance for the switched graph $\hat{\Network}$.
  \label{thm:switch_partition}
\end{theorem}
\begin{proof}
	The imbalance the partition $\sigma$ on $\Network$ is $\sigma \Adjacency \sigma^\T$, and for the switched partition and graph we have
	\begin{align}
		\hat{\sigma} \hat{\Adjacency} \hat{\sigma}^\T = \sigma S S \Adjacency S S \sigma^\T = \sigma \Adjacency \sigma^\T
	\end{align}
	because $S S = I$.
\end{proof}

Even if the partition itself is not balanced, switching is defined for any partition.
If $\Network$ is balanced, we can take the balanced partition $\Nodes_1$ and $\Nodes_2$, in which case all the negative links become positive (because they fall between $\Nodes_1$ and $\Nodes_2$), so that we end up with a completely positive graph.
In reverse, the same thing holds: if we can find a switching $S$ such that $S \Adjacency S$ is completely positive, $\Network$ is balanced, and the switching $S$ defines the optimal partition.
See also~\citet{Hou2003}.
The same principle does not hold for weak structural balance.
For example, a triad with three negative links contains a single one after switching so that the original was weakly balanced but the switched one is not.

When \citet{Abelson1958} introduced the idea of switching, they considered a node with the maximal difference of the positive and negative degree: $d_i^+ - d_i^-$.
Switching the signs of all its links would then decrease the total number of negative links, while the balance would remain unchanged.
The final number of negative links then gives an upper bound on the number of negative links that would need to be removed (or switched) in order to yield structural balance.
In other words, it provides an upper bound on the line index of imbalance.

More recently, a rather similar approach was used by \citet{Iacono2010}. 
They follow the same procedure as \citet{Abelson1958} for reducing the number of negative links to arrive at an upper bound for the line index of imbalance.
The optimal solution may contain even fewer negative links.
\citet{Iacono2010} also provide a way to arrive at a lower bound.
The key idea is to associate each negative link to an edge-independent unbalanced cycle, which is easier if the graph contains few negative links.
This procedure relies on the fundamental cycles we briefly encountered earlier. 
Clearly at least one link must change for each edge-independent unbalanced cycle.
Even though some cut set may reduce the number of negative links, no cut set can reduce it more than the number of the unbalanced edge-independent cycles.
Hence, this provides a lower bound on the line index of imbalance.

\subsection{Weak structural balance}

The previous subsection dealt only with a split in two factions.
We can provide similar definitions for a split in multiple factions.
In particular, the number of inconsistencies with structural balance for a given partition into $\Nodes_1, \Nodes_2, \ldots, \Nodes_q$ is 
\begin{equation}
  C = \frac{1}{2}\sum_{\Nodes_i \neq \Nodes_j} C^+(\Nodes_i, \Nodes_j) + \sum_{\Nodes_i} C^-(\Nodes_i).
\end{equation}

Note that if a network contains only positive or only negative links, the minimum line index of imbalance is, by definition, $0$.
For a network of only positive links, the trivial partition consisting of a single cluster provides such a solution.
Similarly, for a network of only negative links, the trivial partition consisting of each node in its own cluster, commonly called the singleton partition achieves zero imbalance.

However, if all links are negative there is an interesting problem: finding the minimum number of factions required for obtaining an imbalance of $0$.
Having $n$ factions (clusters), with each node in its own faction with an imbalance measure of $0$, most often, has little value.
It is reasonable to think that this measure could be achieved with fewer factions.
For example, for a bipartite graph with all negative links, we have to use only $2$ factions.
This minimum number of factions necessary to obtain an imbalance of $0$ is known also as the chromatic number: the minimum number of colors necessary to color each node such that two nodes that are connected have different colors.
This is a much studied area of research in graph theory. It is an NP-complete problem.
This connection was recognized by~\citet{Cartwright1968}.
The similar problem for positive links is oddly enough trivial: the maximum number of communities for which the imbalance is still $0$ simply corresponds to the connected components.

\subsection{Blockmodelling}

The original block model function proposed by~\citet{Doreian1996a} is exactly equivalent to the line index of imbalance.
They also propose a more general form however, weighting differently positive or negative violations of balance:
\begin{equation}
	C = \alpha C^+ + (1 - \alpha)C^-
\end{equation}
where $\alpha = 0.5$ returns (half) the original line index. 
However, this generality comes with costs. 
Without surprise, different values for $\alpha$ return different values of $C$. 
More consequentially, different partitions of the nodes can be returned. This implies there is no principled way for selecting a value of $\alpha$ and hence a partition.
This issue was noted by \citet{Doreian2015}. It can be called `the alpha problem' which amounts to understanding the interplay of the number of positive and negative links in a signed network, the shape of the criterion function, and the role of $\alpha$ in determining partitions. 

The blockmodeling approach partitions the nodes into positions and the links into blocks which are the sets of links between nodes in the positions. 
There is only one type of blockmodel in accordance with structural balance: positive blocks on the main diagonal and negative blocks off the diagonal. 
Of course, for most empirical situations, the links contributing to the line index for imbalance are distributed across blocks. 
To address this, \citet{Doreian2009} examined other possible blockmodels.
They considered two mutually antagonistic camps being mediated by a third group (either internally negative or not).
So, rather than seeking a blockmodel consisting of diagonal positive blocks and off-diagonal negative blocks, they proposed blockmodels with positive and negative blocks appearing anywhere.
For the empirical networks they studied, the results were better fits to the data, according to the line index, and more useful partitions.
Unfortunately this comes at a price: if the number of clusters is left unspecified  {\it a priori}, the best partition is the singleton partition (i.e. each node in its own cluster). 
This line of research is further studied by~\cite{Figueiredo2013,Brusco2011}.

Stochastic block models can also deal with negative links~\cite{Jiang2015}, but we do not discuss them further here.

\subsection{Community detection}

Assuming structural balance holds for a network, the resulting partition is a set of clusters with primarily positive ties within them. 
Structural balance models would not be informative for networks without negative ties. 
Even so, the positively-connected clusters may contain some further sub structure. 
Most networks that contain only positive links can show a clear group structure, commonly called community structure or modular structure, covered in Chapter~3 in this book.
One of the most popular methods for community detection in networks with only positive links is known as modularity.
It is defined as
\begin{equation}
	\Q = \sum_{ij} \left(\Adjacency_{ij} - \frac{d_i d_j}{2m}\right)\delta(\sigma_i, \sigma_j)
\label{eq:modularity}
\end{equation}
where it is assumed that $\Adjacency_{ij}$ only contains positive entries and $\sigma_i$ denotes the community of node $i$ (i.e. if $\sigma_i = c$ it means that node $i$ is in community $c$) and where $\delta(\sigma_i, \sigma_j) = 1$ if $\sigma_i = \sigma_j$ and otherwise $\delta(\sigma_i, \sigma_j) = 0$.
Although this method suffers from a number of problems, most prominently the resolution limit~\cite{Fortunato:2007p183}, it seems to return sensible partitions for graphs with only positive links.

However, modularity suffers from a problem when some of the links are negative \citep{Traag2009,gomez2009}.
In particular, imagine there are two fully connected subgraphs, the first with $n_1 = 5$ nodes and the second with only $n_2 = 2$ nodes while there are $n_1n_2 = 10$ negative links between these two subgraphs.
Using the ordinary definitions, the weighted degree for the first subgraph
would be $d_i = 4 - 2 = 2$ because each node has $4$ links to the other in the subgraph, and $2$ negative links to the other subgraph.
Similarly, the weighted degree for the second subgraph is $d_i = 1 - 5 = -4$,
and the total weight is $m = \binom{5}{2} + \binom{2}{2} - 5\cdot2 = 1$.
Hence, for any link within the first subgraph,
\begin{equation}
  \Adjacency_{ij} - \frac{d_i d_j}{2m} = 1 - \frac{2 \cdot 2}{2} = -1
  \label{eq:clique1}
\end{equation}
and for the second subgraph
\begin{equation}
	\Adjacency_{ij} - \frac{d_i d_j}{2m} = 1 - \frac{(-4) (-4)}{2} = -7
  \label{eq:clique2}
\end{equation}
and for any link in between
\begin{equation}
	\Adjacency_{ij} - \frac{d_i d_j}{2m} = 1 - \frac{(2) (-4)}{2} = 3.
  \label{eq:clique1_2}
\end{equation}
This is rather surprising, as it says that two nodes that are positively connected should be split apart (their contribution is negative), while two negatively connected nodes should be kept together (their contribution is positive).
Of course the correct partition here should be a partition into two communities: all nodes of the first subgraph forms one community and all nodes of the second subgraph forms the other.
However, summing up the contributions in Eq.~(\ref{eq:clique1}) and Eq.~(\ref{eq:clique2}) the quality of such a partition would be 
\begin{equation}
	\frac{5 \cdot 4}{2} \cdot (-1) + \frac{2 \cdot 1}{2} \cdot (-7) = -17
\end{equation}
while if there is only one single large partition, adding the contribution from Eq.~(\ref{eq:clique1_2}), we obtain
\begin{equation}
	-17 + 5 \cdot 2 \cdot 3 = 13.
\end{equation}
In short, modularity cannot be simply applied to signed networks, as the results are inconsistent with the correct partition (cf.~\cite{Doreian2015}).
Hence, modularity needs to be corrected in some way to account for the presence of negative links for it to be useful for signed networks.

Consistent with structural balance, we would expect negative links to be between communities, while the positive links are within communities.
Hence, if we define the quality of the partition on the positive part as
\begin{equation}
	\Q^+ = \sum_{ij} \left(\Adjacency^+_{ij} - \frac{d^+_i d^+_j}{2m}\right)\delta(\sigma_i, \sigma_j)
\end{equation}
and on the negative part as
\begin{equation}
	\Q^- = \sum_{ij} \left(\Adjacency^-_{ij} - \frac{d^-_i d^-_j}{2m}\right)\delta(\sigma_i, \sigma_j)
\end{equation}
then we should like to maximize $\Q^+$ and minimize $\Q^-$.
We can do so by combining $\Q = \Q^+ - \Q^-$ which then becomes
\begin{equation}
	\Q = \sum_{ij} \left[\Adjacency_{ij} - \left( \frac{d_i^+ d_j^+}{2m^+} - \frac{d_i^-d^-_j}{2m^-} \right) \right]\delta(\sigma_i, \sigma_j)
\end{equation}
where $\Adjacency_{ij} = \Adjacency^+_{ij} - \Adjacency^-_{ij}$ as throughout this chapter.
In essence, this comes down to using a null-model that is adapted to signed
networks.
More details can be found in \citet{Traag2013}, Chapter 5.

More generally speaking, one could always define $\Q^+$ for a partition on the positive subnetwork and $\Q^-$ on the negative subnetwork and then define a new quality function as $\Q = \Q^+ - \Q^-$.
For some methods this turns out to give quite nice results, for example for the Constant Potts Model (CPM)~\cite{Traag2011}.
This method was introduced to circumvent any particular form of the resolution limit.
The formulation (again assuming $\Adjacency_{ij}$ is only positive) is simple:
\begin{equation}
	\Q = \sum_{ij} [\Adjacency_{ij} - \gamma]\delta(\sigma_i,\sigma_j).
\end{equation}
Here $\gamma$ plays the role of a resolution parameter, which needs to be chosen in some way.
This parameter has a nice interpretation though, which could motivate a particular parameter setting, and functions as a sort of threshold.
In any optimal partition, the density between any two communities is no higher than $\gamma$, i.e. $e_{cd} \leq \gamma n_c n_d$ where $e_{cd}$ is the number of edges between $c$ and $d$ and $n_c$ and $n_d$ the number of nodes in that community.
Similarly, any community has a density of at least $\gamma$, i.e. $e_{cc} \geq \gamma\binom{n_c}{2}$.
Even stronger, in fact, any subset of a community is connected to the rest of its community with a density of at least $\gamma$ in an optimal partition.

If we extend our previous suggestion of combining the positive and negative parts we arrive at the following:
\begin{align}
	\Q =& \Q^+ - \Q^- \\
	   =& \sum_{ij} [\Adjacency^+_{ij} - \gamma^+]\delta(\sigma_i, \sigma_j) - \nonumber \\
      & \sum_{ij} [\Adjacency^-_{ij} - \gamma^-]\delta(\sigma_i, \sigma_j) \\
		 =& \sum_{ij} [(\Adjacency^+_{ij} - \Adjacency^-_{ij}) - (\gamma^+ - \gamma^-)]\delta(\sigma_i, \sigma_j)
\end{align}
which, by setting $\gamma = \gamma^+ - \gamma^-$, leads to
\begin{equation}
	\Q = \sum_{ij} [\Adjacency_{ij} - \gamma]\delta(\sigma_i, \sigma_j).
\end{equation}
In other words, for CPM, there is no need to treat negative links separately, and we can immediately apply the same method.

Finally, for $\gamma = 0$, CPM is equivalent to optimizing the line index of imbalance.
Indeed, note that we can write the line index of imbalance as
\begin{equation}
  C = \frac{1}{2}\sum_{ij} \left[ \Adjacency^-_{ij} \delta(\sigma_i, \sigma_j) + \Adjacency^+_{ij}(1 - \delta(\sigma_i, \sigma_j)) \right] 
\end{equation}
We can rewrite this as
\begin{align}
	C &= \frac{1}{2} \sum_{ij} \left[\Adjacency^-_{ij} \delta(\sigma_i, \sigma_j) + \Adjacency^+_{ij}(1 - \delta(\sigma_i, \sigma_j))\right] \\
	  &= \frac{1}{2} \sum_{ij} \left[(\Adjacency^-_{ij} - \Adjacency^+_{ij}) \delta(\sigma_i, \sigma_j) + \Adjacency^+_{ij}\right] \\
		&= m^+ - \frac{1}{2}\sum_{ij} \Adjacency_{ij} \delta(\sigma_i, \sigma_j).
		\label{equ:imbalance_eq_comm}
\end{align}
so that $C = m^+ - -\frac{1}{2} \Q$ for the CPM definition of $\Q$.

Given any particular quality function, the problem is always how to find a particular partition that maximizes this quality function.
In general, this problem cannot be solved efficiently (it is NP-hard), and so we have to employ heuristics.
One of the best performing algorithms for optimizing modularity is the so-called Louvain algorithm~\cite{Blondel2008}. It can be adapted for taking into account negative links.
In addition, it can also be adapted for CPM (and other quality functions still).
See \url{https://pypi.python.org/pypi/louvain/} for a \texttt{Python} implementation designed for handling negative links and working with these various methods.

We do not discuss in detail how the algorithm works, but do discuss one particular element that needs to be changed for dealing with negative links.
The basic ingredient of the algorithm is that it moves nodes to the best possible community.
Ordinarily, in community detection, all communities are connected, and hence, the algorithm only needs to consider moving nodes to neighboring communities.
However, this property no longer holds when negative links are present.
A trivial example is a fully connected bipartite graph with all negative links.
In that case, none of the nodes in any community are connected at all.
When only considering neighboring communities, the algorithm never considers moving a node to a community to which it is not connected.
In the end, if the algorithm starts from a singleton partition (i.e. each node in its own community), it will remain there.
So, we need to calculate the change in $\Q$ for all communities, even if it is not connected to that community.
Unfortunately this increases the computational time required for running the algorithm.
Nonetheless, the algorithm is quite fast.
Of course, it only provides a lower bound on the optimal quality value.
Hence, for minimizing the line index of imbalance it only provides an upper bound.

\subsubsection{Temporal community detection}

One concern when studying the evolution of balance is that we also would like to track the partition over time.
For example, if we have two network snapshots and we try to detect the partition minimizing the imbalance, there is an arbitrary assignment to the clusters $-1$ and $1$ (or $0$ and $1$) in the sense that simply relabeling the partition by exchanging the $-1$ and $1$ yields exactly the same imbalance.
For two communities this is still reasonably limited, but for more communities the problem may become more difficult, especially when dealing with many snapshots throughout time.

We rely on a method introduced by~\citet{Mucha2010} to do temporal community detection, while still accounting for negative links.
Because this is not the core issue in this chapter, we discuss it only briefly.
The idea is to create one large network, which contains all the snapshots of the same network.
Then, each node represents a temporal node: a combination of a time snapshot and the original node.
Without any links between the different snapshots, the large network would thus consist of as many connected components as there are snapshots (assuming each snapshot is connected).
Each snapshot is commonly called a slice, and each link within a slice is called an intraslice link.
We introduce additional interslice links, which connects two identical nodes in two consecutive time slices (i.e. they represent the same underlying node, but at a different time) with a certain strength, called the interslice coupling strength.
This requires also some additional changes on the Louvain algorithm.

\section{Empirical analysis}

\noindent Empirical research has shown that while few empirical networks are close to balance, at least  they are much closer than can be expected at random.
Hence, there is considerable evidence that structural balance holds to some extent, at least for weak balance. The evidence for strong structural balance is far more modest with many exceptions present in the literature.
In particular, the all negative triad was found relatively frequently by \citet{Szell2010} contradicting strong structural balance.
They found triads having a single negative link (which is the only triad that is weakly unbalanced) much more rarely.
Overall, their evidence favors weak structural balance over strong structural balance.
Contrary to dynamical models of sign change, they find that links almost never change sign.
However, there is relatively little research into the dynamics of structural balance. Examples where this has been done include \citet{Hummon2003,Doreian2001,Marvel2011}, and \citet{traag_dynamical_2013}.

We here briefly investigate the dynamics of the network of international relations, where structural balance is argued to play a role by \citet{Doreian2015}.
We gathered data from the Correlates of War\footnote{\url{http://www.correlatesofwar.org/}} (CoW) dataset, which collects a variety of information about international relations.
We create a signed network based on their latest data on alliances (v4.1), representing the positive links, and the militarized interstate disputes (MID, v4.1), representing the negative links.
To arrive at a single weight for each link, we sum the different weights on alliances and MIDs for each dyad (a dyad can be involved in multiple alliances and multiple MIDs at the same time).
Each MID generates an undirected (negative) link for all states that are involved on different sides.
For example, if the US and the UK would be in conflict with Egypt and the USSR, then this would generate four negative links: US-Egypt, UK-Egypt, US-USSR and UK-USSR.
The MID weight is set to $\frac{\text{HighestAct}}{21}$ so that the weight is in $[0, 1]$ (see CoW documentation for more details).
Each alliance generates a link for all dyads involved in the alliance.
The weighting is more complicated, since no {\it a priori} weights are assigned.
We chose to weigh a defense pact with a weight of $\frac{10}{14}$, nonaggression by a weight of $\frac{2}{14}$, and both a neutrality and an entente by a weight of $\frac{1}{14}$.
The single weight is then the sum of the alliance weight minus the sum of the MID weight.
Note that a dyad may be involved in multiple MIDs and/or multiple alliances at the same time, so that the individual weight of a link is not necessarily restricted to $[-1, 1]$. 

We find that structural balance does not follow any singular trend, and certainly does not converge to structural balance and remains stable. The same was found by \citet{Doreian2015} and \citet{Vinogradova2014} where an earlier version of the CoW data was used.
We detect communities using CPM with $\gamma=0$ and use the approach by \citet{Iacono2010}, which we abbreviate as IRSA (after the authors).
We ran the Louvain algorithm for CPM both with unlimited number of communities (corresponding to weak structural balance) and also with the number of communities restricted to two (corresponding to strong structural balance).

\begin{figure}[t]
  \centering
  \includegraphics[width=\linewidth]{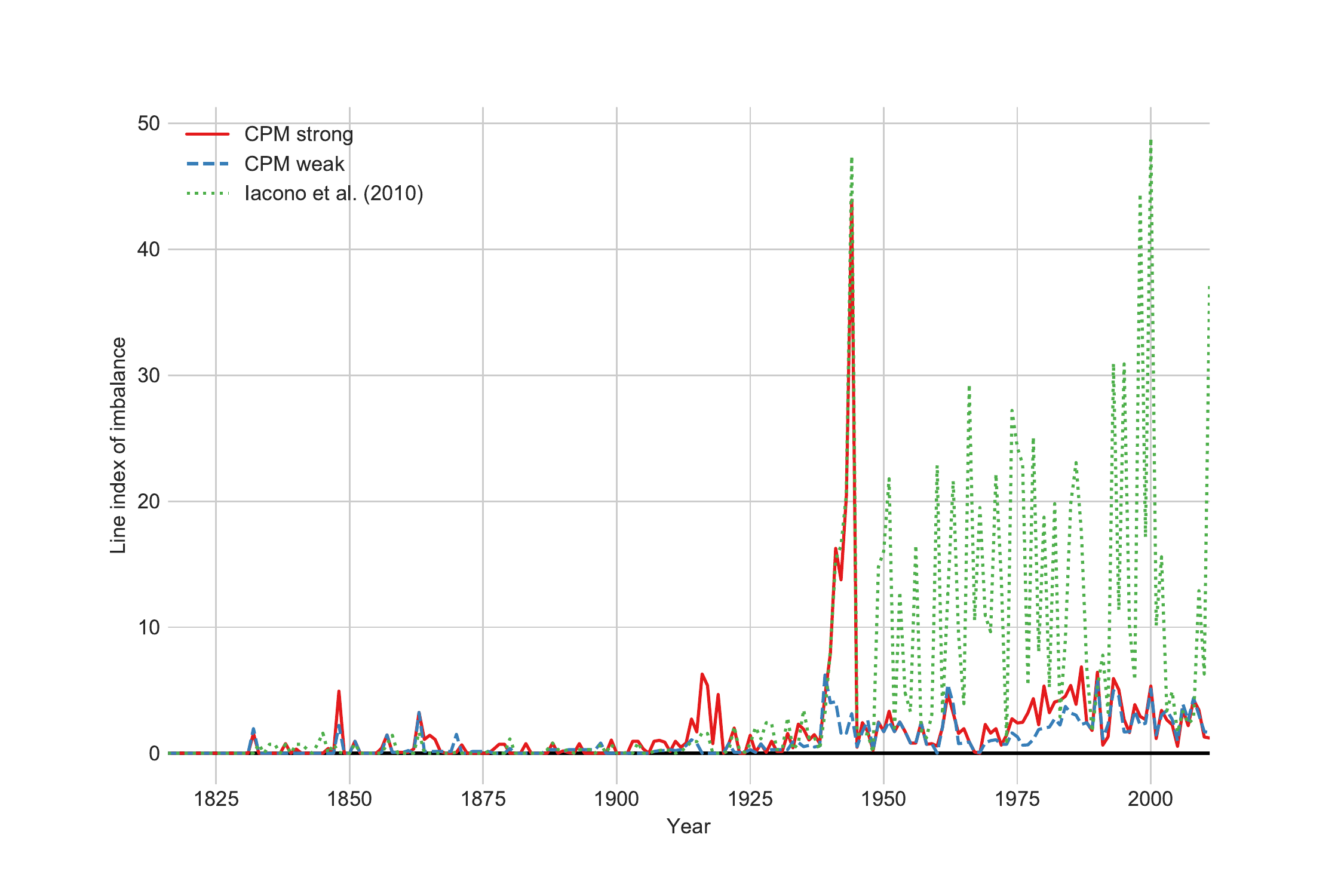}
  \caption{\textbf{Balance timeline}. 
    The line index of imbalance using two different methods. 
    The approach by \citet{Iacono2010} only works for strong balance.
    The CPM approach can be applied both the strong and weak balance.
    CPM seems to provide more stable results than the approach by \citet{Iacono2010}.
  }
  \label{fig:balance_timeline}
\end{figure}

IRSA provides less stable results compared to the CPM estimates (see Fig.~\ref{fig:balance_timeline}).
Perhaps, with more computation time, more accurate results could be achieved.
Even so, regardless of the method, no clear stability emerges.
There are some large peaks of imbalance around WWII, which we discuss.
But during the Cold War, and even after the Cold War, no particular convergence towards 0 imbalance is observed. This is not unreasonable, as the international system is subject to new shocks when new conflicts, some of which are major conflicts, emerge.
Rather than settling at some level of balance,  some unbalance remains in the system which never completely dissipates.
Most often, the difference between strong and weak structural balance usually is not so large.
This implies that a partition of nations into just two factions already explains much of the structure in international relations. At face value,
this suggests that strong balance is at least a reasonable first approximation, and provides some evidence that strong balance is operating in the international system. It is likely that weak balance operates also, perhaps at different timescales.
Nonetheless, there are some clear deviations in the patterns of imbalance.

In particular, both IRSA and CPM find that 1944 shows a large peak with an imbalance of $43.9$ (CPM) or $47.4$ (IRSA), whereas weak structural balance only has an imbalance of $3.14$. For this time point, using weak balance may be more useful. This result is due to the large number of conflicts among various parties, which weak balance can accommodate, but which presents problems for strong balance (see Fig.~\ref{fig:network_1944}).
Indeed, of the 1785 triads in this network, there are 411 strongly unbalanced triads, of which 406 are all-negative triads.
The all-negative triad is considered unbalanced under strong balance, but balanced under weak balance.
This leaves only 5 unbalanced triads under weak balance (although this does not preclude the existence of longer unbalanced cycles).

Many of the all-negative triads are attributable to conflict among nine different countries who were all in conflict with each other: France, Germany, Italy, Hungary, Bulgaria, Romania, Russia, Finland, and New Zealand. 
Many other countries were opposed to at least two others of this large conflict: Japan for example was in conflict with both Russia and New Zealand. These conflicts may be unrelated. But they serve to create an additional unbalanced triad (in the strong sense).
The weakly unbalanced triads involve the UK and Turkey.
The UK was allied with Portugal and Turkey, but Portugal was also allied with Spain (through the alliance between the dictators controlling both countries) which was in conflict with the UK.
Turkey was allied with Germany, Hungary and Iraq in addition to the UK while the UK was in conflict with both Germany and Hungary.
At the same time, Germany was also in conflict with Hungary and Iraq, complicating things further. Clearly, WWII featured many dyadic conflicts, each with their own dynamics.

\begin{figure}[t]
  \centering
  \includegraphics[width=0.8\linewidth]{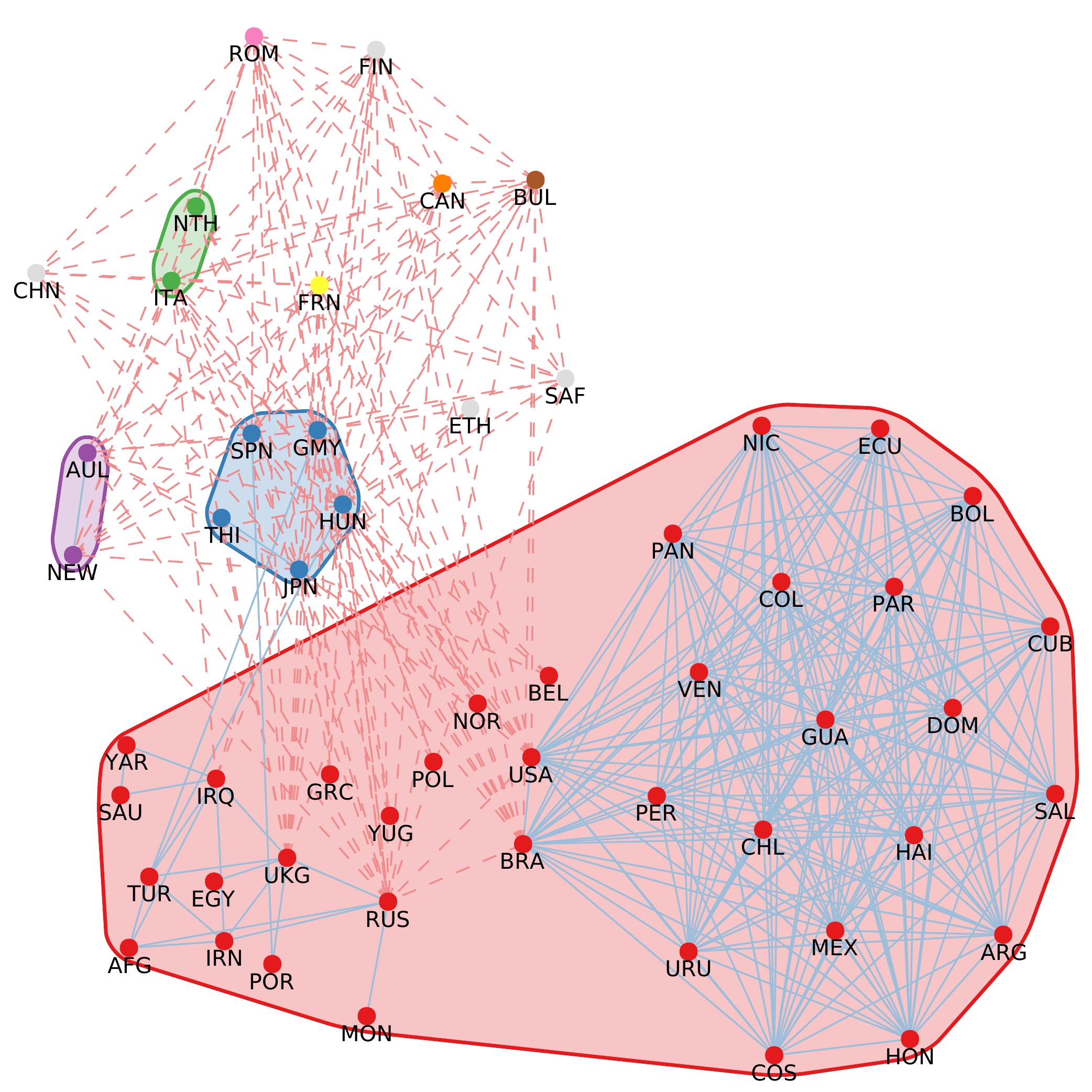}
  \caption{\textbf{International Relations 1944}. 
    The solid lines represent positive links and the dashed lines represent negative links.
    The countries that are clustered together are encircled.
  }
  \label{fig:network_1944}
\end{figure}

There is another interesting observation: weak structural balance is higher than strong structural balance for 1939.
This should not be the case ordinarily, as the minimal imbalance in weak structural balance should always be lower than strong structural balance.
This then seems due to the shift of alliances during WWII.
Since the clustering also favors a certain continuity over time, it may be better to cluster countries in a more stable way, without accounting for short term deviations.
This is what seems to happen in 1939.
In particular, Russia was still allied with Germany and Italy, while Russia is
in conflict with the UK, France and Belgium at that time.
Similarly, Hungary is allied with Turkey, and Spain with Portugal.
Surprisingly, the UK also had some conflict with the USA at that time according to the CoW data.

\begin{figure}[t]
  \centering
  \includegraphics[width=0.8\linewidth]{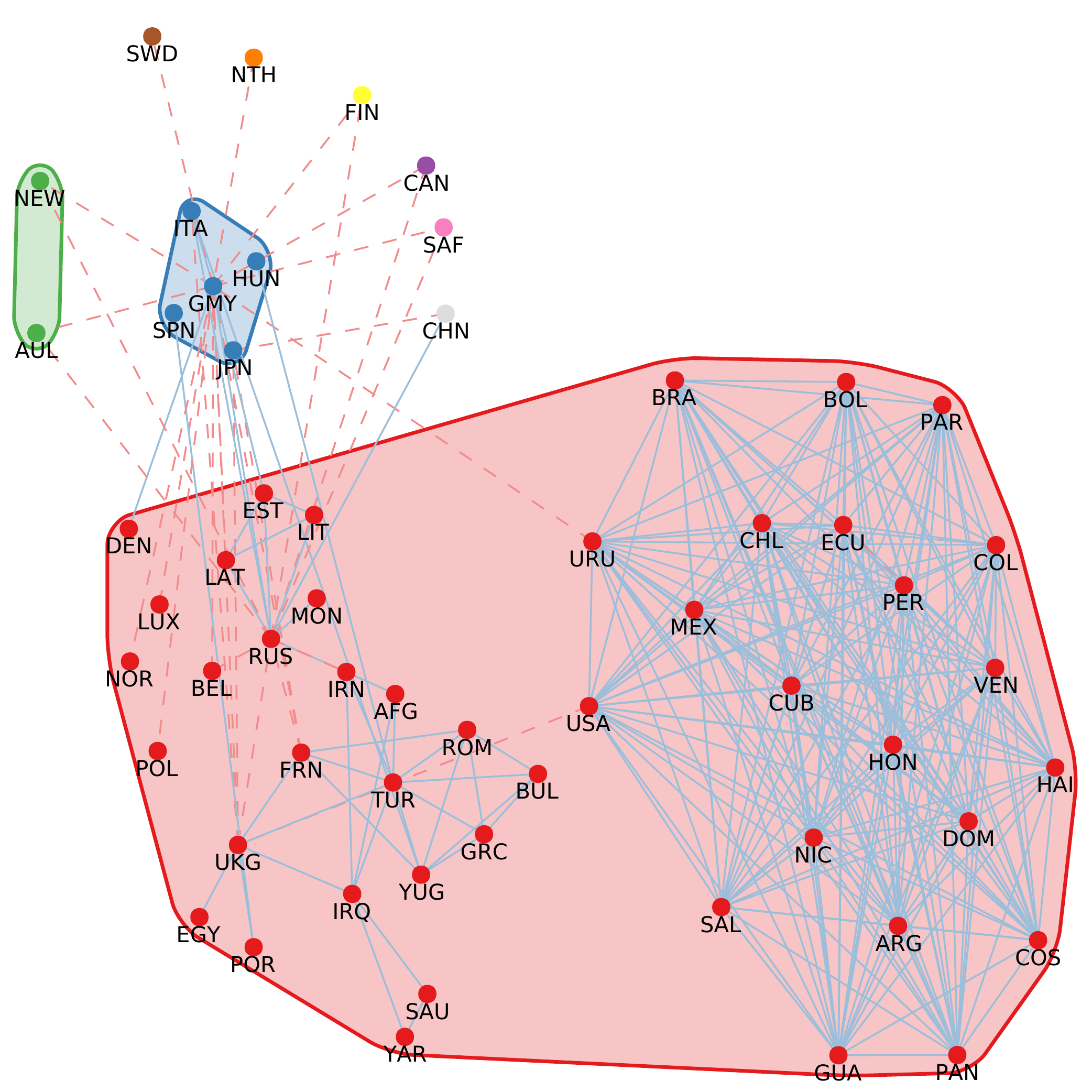}
  \caption{\textbf{International relations 1939}. 
    The solid lines represent positive links and the dashed lines represent negative links.
    The countries that are clustered together are encircled.
  }
  \label{fig:network_1939}
\end{figure}

At the height of the Cold War, we see the familiar division (see Fig.~\ref{fig:map_1962}).
We also see the non-aligned states clustered outside of the familiar division.
Yet some countries are clustered differently than what one would expect.
For example, much of the Arab world is clustered with the West because of the alliance of Morocco and Libya with France. But note that Algeria is not as it was fighting a war of independence with France.
Also, Yugoslavia is commonly seen as non-aligned in the Cold War, but here it is clustered with the West through its alliances with Greece and Turkey. 

Finally, in more recent times the weak balance clustering seems increasingly unrealistic.
This is due to the fact that even if some countries are only weakly positively connected, they are immediately considered as a single cluster. 
In 2010 for example, most of the world is grouped together in a single cluster, except Africa and some exceptions. 
Nonetheless, some clearly separate clusters exist. 
We therefore also detected clusters using CPM with $\gamma=0.1$ for 2010. 
The results are shown in Fig.~\ref{fig:map_year=2010_gamma=0.1}. 
There are clearly different clusters in Africa, something missed completely when partitioning with weak structural balance. 
Africa is divided into a Central African bloc, a Western African bloc, and a Northern African bloc clustered with Arab nations in the Middle East, with the remainder of Africa scattered across other communities. 
The former USSR remains a separate community. 
The so-called West breaks into two communities: North and South America constitute a community whereas Europe becomes a separate community. 

\begin{figure}[t]
  \centering
  \includegraphics[width=\linewidth]{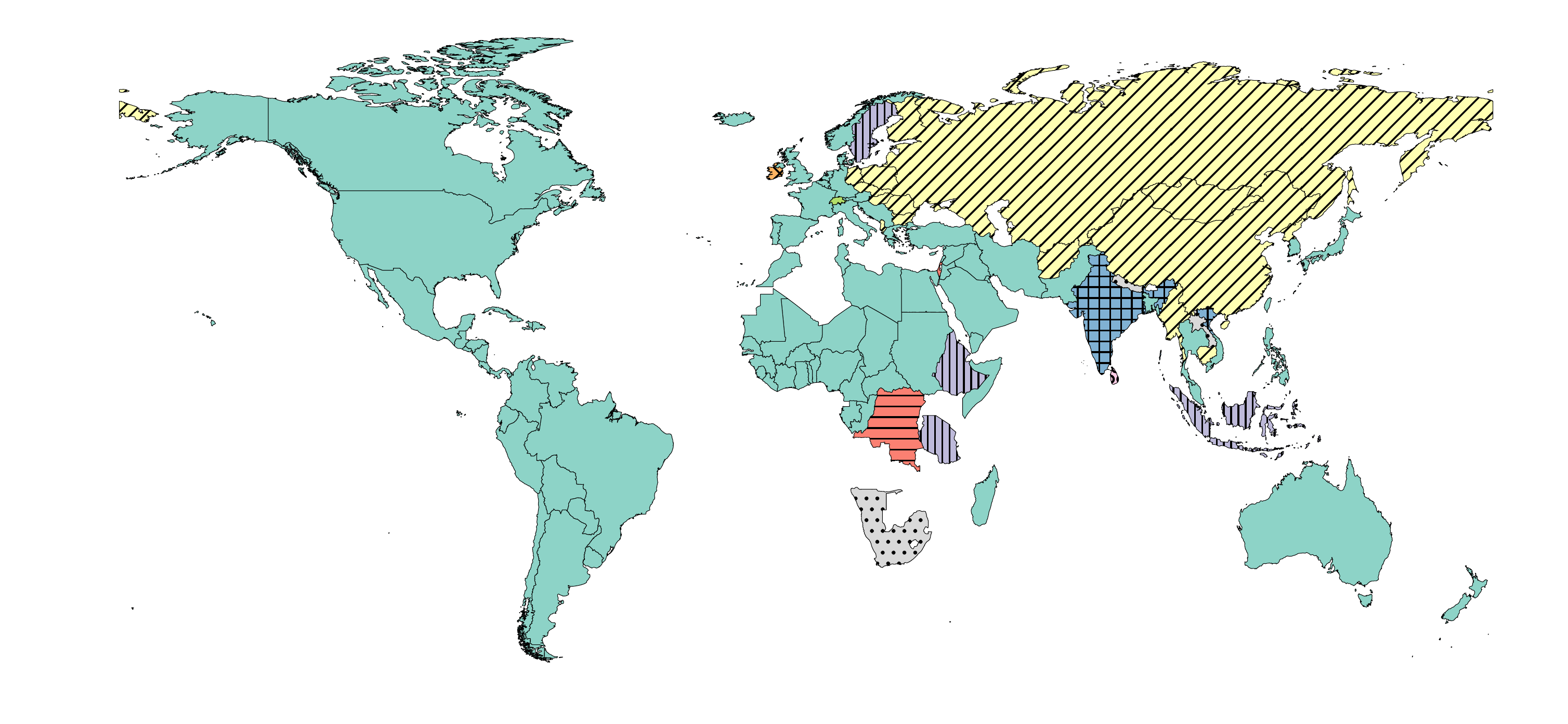}
  \caption{\textbf{Map of weak balance partition in 1962}.
  }
  \label{fig:map_1962}
\end{figure}

\begin{figure}[t]
  \centering
  \includegraphics[width=\linewidth]{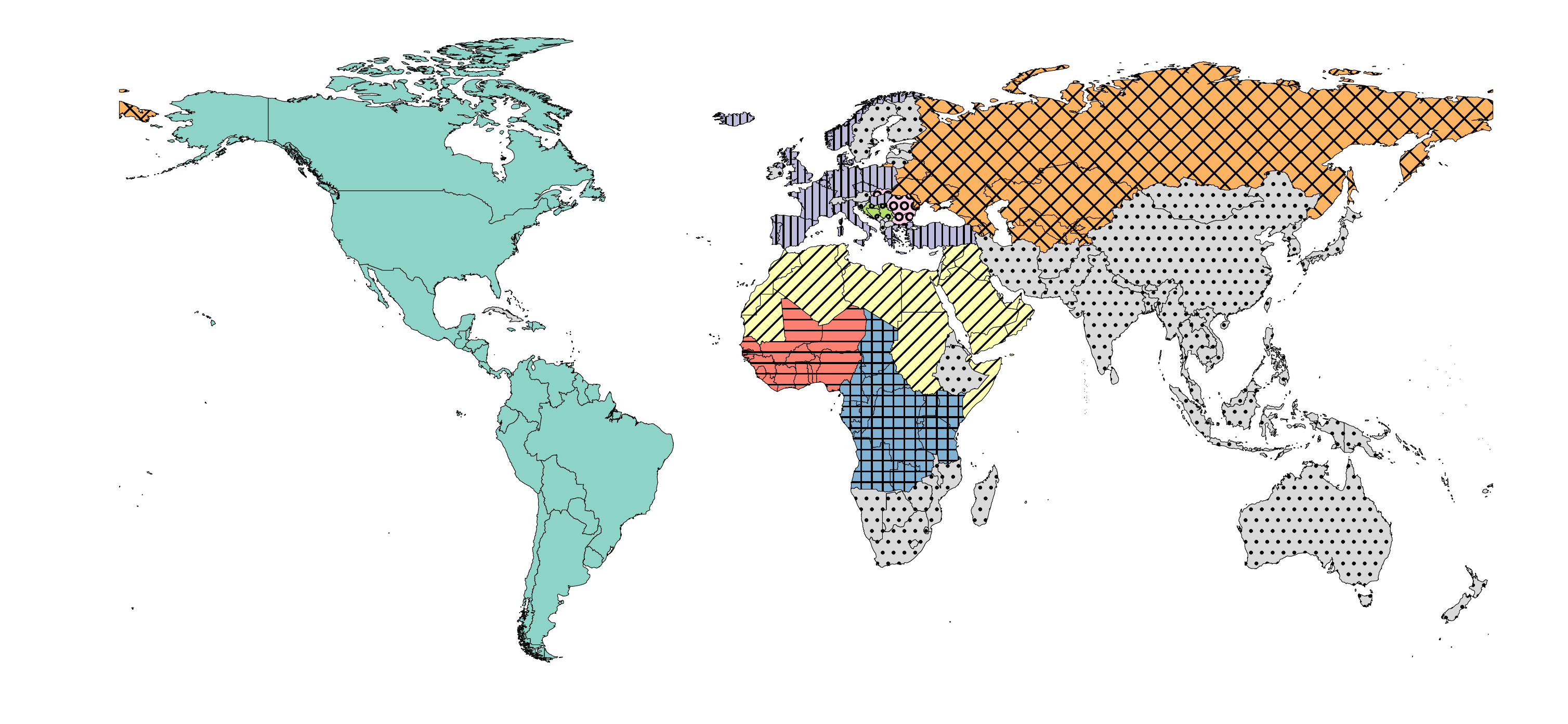}
  \caption{\textbf{Map of CPM partition with $\gamma = 0.1$ in 2010}.}
  \label{fig:map_year=2010_gamma=0.1}
\end{figure}

This is also interesting from another perspective.
Structural balance emphasizes both that negative links ought not exist within clusters and positive links ought not exist between clusters. This seems too restrictive by ignoring the presence of some conflict within clusters along with positive ties between clusters.
Arguably, it makes more sense to allow for a few positive links between clusters without requiring them to be considered immediately a single cluster.
Indeed, when using CPM with $\gamma = 0.1$ relatively less conflict happens within clusters, and most conflict takes place between clusters. 
Nonetheless, strong balance remains a reasonable first approximation. 

\section{Summary and Future Work}

\noindent Partitioning signed networks raises methodological issues that differ from those involved in partitioning unsigned networks. 
Various approaches have been developed. 
We started our discussion with a consideration of structural balance as it provides a substantively driven framework for considering signed networks. 
Formulated in terms of exact balance, the initial results in the literature take the form of existence theorems, which we discussed in some detail. 
We distinguished strong structural balance and weak structural balance. 
Empirically, most signed networks are not exactly balanced. 
One of the underlying assumptions of classical structural balance theory is that signed networks tends towards balance. 
To assess such a claim, it is necessary to have a measure of the extent to which a network is balanced or imbalanced. 
We discussed some of the measures in the literature but focused primarily on the line index of imbalance. 
Obtaining this measure is an NP-hard problem. 
We provided theorems regarding obtaining this measure and its upper and lower bounds.
		
In discussing strong structural balance, we considered spectral theory and presented some results showing how this is another useful approach for obtaining measures of imbalance. 
In doing so, we revisited the concept of switching. 
For partitioning signed networks, we considered signed blockmodeling as a method, pointing out its value and serious limitations. 
We considered community detection and outlined ways in which is can be adapted usefully to partition signed networks. 
In discussing this we considered also the Constant Potts Model (CPM) and how it can be used to partition signed networks. 
We discussed briefly the notion of temporal community detection.

With the formal  results in place, we turned to an empirical example using data from the Correlates of War (CoW) data. 
We applied two methods to obtain partitions for different points in time. 
We made no attempt to assess which is the `best' partitioning method, for they all have strengths and weaknesses. 
However, we did initiate a discussion regarding the conditions under which one method may perform better than others---without being universally the `best' under all conditions.  
This included a discussion of the utility of weak balance and strong balance, the number of clusters and the temporal dynamics of the empirical network we studied.

Our results, consistent with other results for the CoW networks produced by others, is that, temporally, signed networks can move towards balance at some time points and away from balance at others. 
The assumption that signed networks tend towards balance had unfortunate consequences. 
The more important question, substantively, is simple to state: What are the conditions under which these changes take place? 
To some extent, this mirrors the issue of when some methods work better than others. 
The two are related. 
Together, these issues will form a focus for our future work both analytically and substantively.

\bibliography{bibliography}

\end{document}